\documentclass[journal]{IEEEtran}
\usepackage{amsmath,amssymb,amsthm}
\usepackage[utf8]{inputenc}
\usepackage[T1]{fontenc}
\usepackage{graphicx}
\usepackage{booktabs}
\usepackage{xcolor}

\newcommand{\M}{\mathcal{M}}

\newtheorem{theorem}{Theorem}
\newtheorem{proposition}{Proposition}
\newtheorem{definition}{Definition}
\newtheorem{assumption}{Assumption}
\newtheorem{remark}{Remark}

\begin{document}

\title{A Time-to-Boundary Margin for Transient Stability:\\
Unifying Critical Clearing Time and Operating-Point Drift}

\author{Mari\'an Me\v{s}ter%
\thanks{The author is with the Faculty of Electrical Engineering and Informatics,
Technical University of Ko\v{s}ice, Ko\v{s}ice, Slovakia (e-mail: marian.mester@tuke.sk).}}

\maketitle

\begin{abstract}
The loading margin to voltage collapse---the distance in parameter space to the
closest saddle-node bifurcation---is a standard proximity index for voltage
stability. This paper develops its transient-stability counterpart: a margin $\M$
that measures the time to the synchronism boundary rather than a distance, and
that unifies two limits usually treated separately. The critical clearing time
(CCT) is the fast, fixed-parameter limit; the slow drift of the operating point
toward a static loadability limit is the other. $\M$ is defined as the
first-passage time of the joint state--parameter motion to the survival boundary.
We prove and verify that $\M$ equals the CCT exactly on the one-machine--infinite-bus
reduction (deviation $\le 0.01\%$ across loadings on a published benchmark), establishing a certified
single-machine pillar. Under operating-point drift, $\M$ yields an operational lead
time before faults become unclearable; we take the 28~April~2025 Iberian blackout
timeline as an illustrative time scale for the drift rate. On the New England 39-bus
system, an independent benchmark, the single-machine-equivalent reduction reproduces
the CCT within $1.8$--$6.0\%$ (conservatively), and a critical slowing-down signature
flags proximity to the boundary. For the multimachine case we characterize the limits
explicitly: the transfer-conductance work is tightly boundable, while the controlling
unstable equilibrium is the binding obstruction to a certified margin.
\end{abstract}


\section{Introduction}
\IEEEPARstart{T}{ransient} stability---the ability of synchronous machines to retain
synchronism after a large disturbance---remains among the binding constraints on how
close to its limits a power system may be operated~\cite{kundur,ieee_def}. Its classical
index is the critical clearing time (CCT): the longest a given fault may persist before
the post-fault trajectory leaves the region of attraction. The CCT answers a sharp
question, but only at a fixed operating point and for a prescribed fault.

Voltage stability faced an analogous need and answered it with a margin. The loading
margin to voltage collapse---the distance in load-power space to the closest saddle-node
bifurcation---became a standard proximity index because it summarises ``how far from the
boundary'' in one operational number, with sensitivities that guide control, and without
requiring the pattern of load increase to be predicted~\cite{dobson_lu}. No equivalent
time-based margin exists for transient stability: the CCT is a threshold, not a distance
to a moving boundary, and recent work has extended it mainly through sensitivities of the
CCT to parameters~\cite{sharma,mishra} or through simulation-free bounds on the
CCT~\cite{vu_turitsyn,roberts} rather than through a margin that follows the operating
point as it drifts.

This positions $\M$ relative to the established families of transient-stability
assessment. Direct energy-function methods and their controlling-unstable-equilibrium and
boundary-controlling-UEP (BCU) refinements~\cite{athay,pavella,chiang_bcu} certify
stability at a fixed operating point by comparing a transient energy to a critical value;
Lyapunov-function constructions~\cite{bretas,vu_family,anghel,qiu} and network-theoretic
synchronisation results~\cite{dorfler} estimate the region of attraction or synchronisation
conditions directly. All of these return a verdict---stable or unstable---or a region
estimate at a frozen operating point; none reports a time remaining to the boundary under a
moving operating point, which is the quantity $\M$ adds. A second, fast-growing family
replaces analysis with data-driven prediction: learned Lyapunov functions, graph models
and reinforcement-learning controllers~\cite{zhao_nlc,zhao_graph,zhao_rl} predict stability
or margins from simulation data. A parallel line addresses the transient stability of
inverter-based and grid-forming resources~\cite{he_geng,rokrok}. These are complementary
tools for fast online screening, but they yield a black-box estimate tied to a training
distribution rather than the interpretable structure---two explicit boundaries and a
critical-slowing signature---that $\M$ exposes; indeed $\M$ is a deterministic target that
such predictors could be trained to estimate. Finally, the resilience-metric literature
quantifies the impact and recovery of high-impact events~\cite{panteli}; $\M$ is a stability
margin, not a trapezoid-style resilience metric, and we keep that distinction explicit
(Section~\ref{sec:discussion}).

This gap matters because real cascades violate the fixed-point assumption. In the
28~April~2025 Iberian blackout, the operating point drifted from a secure state toward
collapse over tens of seconds before synchronism was lost~\cite{ics_factual,ics_final}; a
fixed-parameter CCT would have signalled nothing until the boundary was already reached.
What is missing is a margin that measures the time remaining to the boundary under the joint
motion of state and parameters.

The grid is shifting to low-inertia, inverter-dominated operation: as synchronous machines
are displaced, the system loses not only inertia and short-circuit strength but also the
dispatchable spinning reserve and ramping flexibility with which operators steer the
operating point back from its limits~\cite{he_geng,panteli}. Transient-stability margins
therefore shrink while the means to restore them thin out, and systems are increasingly
operated close to---occasionally beyond---their secure limits during high-renewable,
low-demand periods. The operating point drifts faster and dwells longer near the boundary,
yet operators screen security from snapshots of the present state; recent events show the
fast, fault-driven outcomes that fixed-point screening misses. Real-time dynamic security
assessment already runs in control rooms---hundreds of $N\!-\!1$ contingencies every few
minutes~\cite{diao}---but the transient-stability verdict it returns is a green/red threshold
at the current point, not a time. The practical aim of $\M$ is to supply that missing
quantity: the lead time remaining before a presently clearable contingency becomes
unclearable under the drift. Its value grows as rotating reserve declines---the scarcer and
slower the corrective flexibility, the earlier it must be triggered---so that preventive
control can be scheduled before the margin is gone rather than after a snapshot turns red; we
treat the classical angle-transient mechanism, with converter-dominated boundaries a separate
question (Section~\ref{sec:conclusion}).

The present reach of this work is bounded. $\M$ is exact on the certified single-machine
pillar and is available as a fast, conservative estimate on multimachine systems through the
single-machine-equivalent reduction---within $6\%$ on the 39-bus benchmark
(Section~\ref{sec:validation}); the raw lossless energy margin, by contrast, is
non-conservative and is not the operational vehicle. A fully certified and simulation-free
multimachine margin remains open (Section~\ref{sec:conclusion}). The contribution is the
definition, its certified anchor, and a map of what an operational tool would require.

This paper introduces such a margin, $\M$, defined as the first-passage time of the joint
state--parameter trajectory to the survival boundary $\Sigma$. It is the transient-stability
sibling of the loading margin: where the loading margin measures distance to a static fold,
$\M$ measures time to the transient boundary, and in doing so unifies the fast limit (the
CCT) and the slow limit (drift toward the static loadability fold) in a single quantity with
the dimension of time. The contributions are:

\begin{enumerate}
\item A time-to-boundary margin $\M$ for transient stability (Section~\ref{sec:margin}),
positioned as the temporal counterpart of the loading margin and reducing to the CCT at a
fixed operating point.
\item An exact single-machine pillar (Section~\ref{sec:pillar}): on the
one-machine-infinite-bus reduction $\M$ equals the CCT exactly, proved and verified to
$\le 0.01\%$ across loadings.
\item A drift formulation (Section~\ref{sec:margin}): two boundaries---dynamic (clearing
capability exhausted) and static (saddle-node)---bound the survival set, and $\M$ gives an
operational lead time, with the drift rate anchored to the Iberian timeline.
\item Independent benchmark validation (Section~\ref{sec:validation}) on the New England
39-bus system: a self-consistent classical foundation, a single-machine-equivalent reduction
conservative to within $6\%$, a critical slowing-down early-warning signature, and an
operational drift scenario in which a contingency that screens green at the present state
($\mathrm{CCT}\approx 425$~ms) becomes unclearable under drift, with $\M$ supplying the lead
time a fixed-point snapshot does not.
\item A delimited multimachine boundary (Section~\ref{sec:multimachine}): the
transfer-conductance work is tightly boundable over the fault-on window, whereas the
controlling unstable equilibrium is identified as the binding obstruction to a certified
multimachine margin.
\end{enumerate}

The remainder is organised as follows. Section~\ref{sec:margin} defines $\M$ and the drift
formulation; Section~\ref{sec:pillar} establishes the single-machine pillar;
Section~\ref{sec:multimachine} characterises the multimachine boundary;
Section~\ref{sec:validation} reports the 39-bus validation; Section~\ref{sec:comparison}
positions $\M$ against the closest analytical CCT metric; Section~\ref{sec:discussion}
discusses the resilience-fragility reading and limitations; Section~\ref{sec:conclusion}
concludes.

\section{The Margin $\M$ and Parameter Drift}\label{sec:margin}

\subsection{System model}
We use the classical multimachine model in the centre-of-inertia (COI) frame. For $n$
machines with constant voltage $E_i\angle\delta_i$ behind transient reactance $x'_{d,i}$,
the network is reduced to the internal nodes, giving a reduced admittance matrix
$Y=G+jB$. The swing dynamics are
\begin{equation}
M_i\dot\omega_i = P_{m,i}-P_{e,i}(\delta),\quad \dot\delta_i=\omega_i,
\label{eq:swing}
\end{equation}
\begin{equation}
P_{e,i}(\delta)=\sum_{j=1}^{n} E_iE_j\bigl(G_{ij}\cos\delta_{ij}+B_{ij}\sin\delta_{ij}\bigr),
\label{eq:Pe}
\end{equation}
where $\delta_{ij}=\delta_i-\delta_j$ and the inertia is $M_i=2H_i/\omega_s$. This
classical model isolates the angle-transient mechanism that $\M$ measures. Inverter-based and
low-inertia operation is the context that makes the margin urgent---faster drift, thinner
reserve---not a claim that $\M$ captures converter dynamics; extending the construction to
systems whose stability boundary is converter-driven is a separate question
(Section~\ref{sec:conclusion}).

\subsection{The survival boundary $\Sigma$}\label{sec:sigma}
We define the survival boundary as the stability boundary of the post-fault equilibrium, made
precise through the region of attraction. Let $x_s$ be an asymptotically stable equilibrium
point (SEP) of the post-fault system~\eqref{eq:swing} and let
$A(x_s)=\{x_0:\varphi(t,x_0)\to x_s \text{ as } t\to\infty\}$ be its region of attraction,
where $\varphi(t,\cdot)$ is the flow of~\eqref{eq:swing}; region-of-attraction estimates are
central also to converter stability analysis~\cite{li_doa}.

\begin{definition}[Survival boundary]
The survival boundary is the topological boundary of the region of attraction,
$\Sigma:=\partial A(x_s)$.
\end{definition}

$\Sigma$ is closed and invariant under~\eqref{eq:swing}. Under the standard genericity
conditions for power-system stability boundaries~\cite{chiang_hirsch,khalil}---(A1) all
equilibria on $\Sigma$ are hyperbolic; (A2) the stable and unstable manifolds of equilibria on
$\Sigma$ intersect transversally; (A3) every trajectory on $\Sigma$ approaches one of these
equilibria as $t\to\infty$---the boundary is the union of the stable manifolds of the unstable
equilibria (UEPs) lying on it,
\begin{equation}
\Sigma=\bigcup_{x_k^u\in\partial A} W^s(x_k^u).
\label{eq:sigma_union}
\end{equation}
Away from the UEPs themselves, $\Sigma$ is a $C^1$ codimension-one manifold, so locally
$\Sigma=\{x:g(x)=0\}$ for a smooth $g$ with $\nabla g\neq 0$; this local defining function is
used in Proposition~\ref{prop:wellposed}. The energy method approximates $\Sigma$ by the
connected component of the level set $\{V(x)=V_{cr}\}$ through the controlling UEP;
Section~\ref{sec:pillar} shows this approximation is exact on the OMIB and
Section~\ref{sec:multimachine} quantifies its error in the multimachine case.

\subsection{Definition of $\M$}
The critical clearing time (CCT) answers a single question at a fixed operating point: how long
may a given fault persist before the post-fault trajectory leaves the survival set? Real
cascades violate the fixed-point assumption---the operating point itself moves while the
contingency unfolds. We therefore define a margin that measures the time to the boundary under
the joint motion of state and parameters.

\begin{definition}[Time-to-boundary margin]\label{def:M}
Let $x=(\delta,\omega)$ be the dynamic state and $p(t)$ a slowly varying parameter vector
(loading, topology). The margin is
\begin{equation}
\M=\min\{\,t\ge 0:(x(t),p(t))\in\Sigma\,\},
\label{eq:Mdef}
\end{equation}
the first time the joint state--parameter trajectory reaches the survival boundary $\Sigma$. Its
dimension is~$[\mathrm{s}]$.
\end{definition}

CCT is the special case of~\eqref{eq:Mdef} with $\dot p=0$ and $x$ driven by a prescribed fault.
The novelty of $\M$ over CCT is the explicit treatment of the parameter motion $p(t)$.

We first establish that~\eqref{eq:Mdef} is well posed. Write the joint motion as
$z=(x,p)\in\mathbb{R}^{2n}\times\mathbb{R}^m$ with $\dot z=F(z)$, where $F$ collects the swing
dynamics~\eqref{eq:swing} and the slow drift $\dot p$. Let $\widehat\Sigma=\{z:x\in\Sigma(p)\}$
be the survival boundary lifted to the joint space.

\begin{proposition}[Existence and regularity of $\M$]\label{prop:wellposed}
Assume $F$ is locally Lipschitz, so the joint trajectory $z(t)=\varphi(t,z_0)$ exists, is unique
and continuous. Then:
\begin{enumerate}
\item[(i)] \emph{(Well-posedness.)} $\M$ in~\eqref{eq:Mdef} is well defined as a value in
$[0,\infty]$, with $\M=\infty$ iff the trajectory never meets $\widehat\Sigma$. If the trajectory
meets $\widehat\Sigma$, the infimum is attained and $\M$ is a genuine minimum (first-passage time).
\item[(ii)] \emph{(Regularity.)} If at $t=\M<\infty$ the trajectory crosses $\widehat\Sigma$
transversally, i.e. $\nabla g(x(\M))^\top\dot x(\M)\neq 0$ for the local defining function $g$ of
Section~\ref{sec:sigma}, then $\M$ is a locally unique, $C^1$ function of $(z_0,\rho)$.
\end{enumerate}
\end{proposition}

\begin{proof}
(i) $\widehat\Sigma$ is closed (Def.~1) and $z(\cdot)$ is continuous, so the hitting set
$T=\{t\ge 0:z(t)\in\widehat\Sigma\}$ is closed and bounded below by $0$. If $T\neq\emptyset$ a
closed set bounded below attains its infimum, so $\M=\min T$ exists; otherwise
$\M=\inf\emptyset=\infty$. (ii) Near $x(\M)$ write $\widehat\Sigma=\{g=0\}$ with $\nabla g\neq 0$.
The scalar map $h(t,z_0,\rho)=g(x(t;z_0,\rho))$ is $C^1$ and $h(\M,z_0,\rho)=0$; transversality
gives $\partial_t h=\nabla g^\top\dot x\neq 0$, so the implicit function theorem yields a unique
$C^1$ solution $t=\M(z_0,\rho)$ in a neighbourhood.
\end{proof}

Transversality is the condition under which $\M$ is a usable margin: it guarantees the
first-passage time varies smoothly with operating point and drift rate, which is what makes
sensitivities $\partial\M/\partial(\cdot)$ meaningful. A tangential (non-transversal) crossing is
the degenerate case in which the trajectory grazes $\Sigma$; it coincides with the critical (fold)
configuration analysed for the OMIB in Section~\ref{sec:pillar}.

\subsection{The single-machine case is exact}
For the one-machine-infinite-bus (OMIB) reduction the potential energy is a single-variable
integral and is therefore exact (no path-dependence), and the controlling unstable equilibrium
$\delta_u^*$ is unique. Consequently $\M$ computed from the lossy energy function equals the CCT
obtained by time-domain simulation. On a published OMIB benchmark whose parameters are
fixed independently of any margin computation~\cite{tripathy_swallow,mester_diss}, the two
agree to within $0.01\%$ across three loadings (Table~\ref{tab:omib}). This exact agreement is the certified pillar on
which the rest of the paper builds; the multimachine extension (Section~\ref{sec:multimachine}) is,
by contrast, heuristic, because the controlling unstable equilibrium becomes both ambiguous and a
source of non-conservatism.

\subsection{Parameter drift and the two boundaries}
Let a scalar stress parameter $\lambda$ scale the loading, $P_{m,i}=\lambda P_{m,i}^0$, and let it
drift at rate $\rho$,
\begin{equation}
\lambda(t)=\lambda_0+\rho t,\qquad \rho=\dot\lambda.
\label{eq:drift}
\end{equation}
Two boundaries bound the survival set:
\begin{itemize}
\item the dynamic boundary $\lambda^*$, where the clearing capability is exhausted,
$\mathrm{CCT}(\lambda^*)=t_{\mathrm{clear}}$ (protection time); beyond $\lambda^*$ no realisable
protection clears a fault in time;
\item the static boundary $\lambda_{\mathrm{SN}}$, the saddle-node at which the equilibrium itself
vanishes ($P_m=P_{\max}$)~\cite{dobson_geom}.
\end{itemize}
Whichever is reached first defines $\Sigma$. Equation~\eqref{eq:Mdrift} below is the adiabatic
(frozen-parameter) reduction of the first-passage time~\eqref{eq:Mdef}, valid under an explicit
separation of timescales.

\begin{assumption}[Quasi-static drift]\label{ass:qs}
The drift is slow relative to the swing dynamics: $\varepsilon:=\rho\,T_{sw}\ll 1$, where $T_{sw}$
is the dominant post-fault swing period and $\rho$ is in units of $\lambda$ per second.
\end{assumption}

Under Assumption~\ref{ass:qs} the system~\eqref{eq:swing} with $\lambda=\lambda(t)$ has the
singularly perturbed two-timescale form $\dot x=f(x,\lambda)$ (fast), $\dot\lambda=\rho$ (slow). By
Tikhonov--Fenichel theory~\cite{khalil,kokotovic}, away from the fold the fast state tracks the
frozen-$\lambda$ critical configuration with an $O(\varepsilon)$ deviation, so the boundary in
parameter space is crossed at $\lambda=\min(\lambda^*,\lambda_{\mathrm{SN}})+O(\varepsilon)$. The
first-passage time~\eqref{eq:Mdef} therefore reduces to the time at which the drifting parameter
reaches that frozen boundary,
\begin{equation}
\M=\frac{\min(\lambda^*,\lambda_{\mathrm{SN}})-\lambda_0}{\rho}\,\bigl(1+O(\varepsilon)\bigr),
\label{eq:Mdrift}
\end{equation}
i.e. the relative error of the frozen-boundary formula is $O(\varepsilon)$. For the Iberian
timescale ($\tau\sim 6$--$78$~s, $T_{sw}\sim 1$~s) one has $\varepsilon\sim 10^{-2}$, so
\eqref{eq:Mdrift} is accurate to the percent level. The complementary regime
$\varepsilon=O(1)$---collapses as fast as the swing---is not covered by~\eqref{eq:Mdrift} and
requires the full joint first-passage problem; this limit is stated in
Section~\ref{sec:discussion}. For the base OMIB of Section~\ref{sec:pillar}
(Table~\ref{tab:omib_params}, at a lighter base loading $\lambda_0$; $t_{\mathrm{clear}}=80$~ms) we obtain $\mathrm{CCT}(\lambda_0)=184$~ms, $\lambda^*=1.61$ (dynamic),
and $\lambda_{\mathrm{SN}}=2.11$ (static); the dynamic boundary is reached first. Whereas a CCT
threshold would warn only at the instant $\lambda$ reaches $\lambda^*$---i.e. at the moment of
failure---$\M$ provides advance warning proportional to the cascade rate (Fig.~\ref{fig:drift}):
$33$, $16$, and $6$~s for slow, medium, and fast cascades, respectively.

\begin{figure}[t]
\centering
\includegraphics[width=\columnwidth]{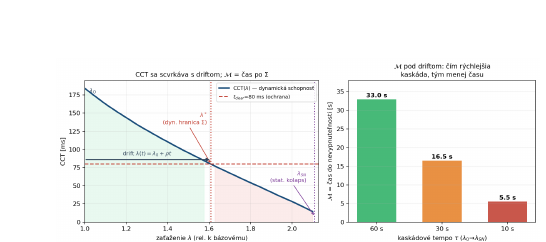}
\caption{Operating-point drift and the time-to-boundary margin. Left: as
$\lambda(t)=\lambda_0+\rho t$ drifts, the clearing-capability boundary $\lambda^*$ (dynamic) is
reached before the saddle-node $\lambda_{\mathrm{SN}}$ (static); $\Sigma$ is set by the first.
Right: $\M$ from~\eqref{eq:Mdrift} is the lead time before faults become unclearable, shown for
slow, medium, and fast drift.}
\label{fig:drift}
\end{figure}

\begin{table}[t]
\centering
\caption{Certified anchor~\cite{tripathy_swallow,mester_diss}: the time-to-boundary
margin $\M$ reproduces the step-by-step CCT across three loadings. Angles in
degrees, times in ms.}
\label{tab:omib}
\footnotesize
\setlength{\tabcolsep}{5pt}
\begin{tabular}{lccccc}
\toprule
Loading & $\delta_0$ & $\delta_c$ & CCT$_{\text{ref}}$ & $\M$ & dev. \\
\midrule
$0.8\,S_n$ & $27.67$ & $73.89$ & $205.6$ & $205.6$ & $+0.01\%$ \\
$S_n$      & $31.47$ & $67.06$ & $180.4$ & $180.4$ & $+0.01\%$ \\
$1.2\,S_n$ & $35.40$ & $61.24$ & $153.7$ & $153.7$ & $+0.01\%$ \\
\bottomrule
\end{tabular}
\end{table}

\subsection{Physical anchoring of $\rho$}
The drift rate is not universal; it is the margin-erosion rate of a specific cascade. The Iberian
blackout of 28~April~2025~\cite{ics_factual,ics_final} anchors it empirically. After precursor
inter-area oscillations ($0.63$~Hz, and $0.21$~Hz between 12:19 and 12:22), the operating point
drifted from a secure state (12:32:00, $\sim 418$~kV at the 400~kV pilot nodes) toward collapse:
the voltage rose to $435.4$~kV by 12:33:16 ($\approx 0.9$~kV/s), synchronism with the Continental
European system was lost at 12:33:19, and full separation completed at 12:33:24. The operating point
thus traversed from secure to collapse over $\sim 78$~s, with an active drift phase of $\sim 19$~s
and a final cascade of $\sim 6$~s. This places realistic $\rho$ in the range corresponding to
$\tau\in[6,78]$~s, within which the demonstration values above fall.

\emph{Scope.} The dominant boundary in the Iberian event was voltage (over-voltage and loss of
reactive support), not the angle-inertial boundary of the OMIB demonstration. What the event anchors
is the timescale of cascade drift, i.e. $\rho$; the voltage boundary itself is a sister margin
(Scenario~B, Section~\ref{sec:discussion}) obtained by replacing the energy surface with a
nose-point/Jacobian-distance surface. The chain voltage drift $\to$ boundary crossing $\to$ loss of
synchronism is what~\eqref{eq:Mdef} measures: the time for the drift to erode the margin to zero.

\section{The Single-Machine Pillar: $\M$ is Exact}\label{sec:pillar}
On the one-machine-infinite-bus (OMIB) reduction the margin $\M$ of Definition~\ref{def:M} coincides
exactly with the critical clearing time. In the fixed-parameter limit this equivalence is a
restatement of the equal-area criterion~\cite{athay,pavella}; we present it not as a result in its
own right but as the certified anchor of the construction---the one configuration in which the
general time-to-boundary definition provably reduces to the classical CCT. Its role is twofold: it
fixes the single case in which $\M$ is exact, against which the multimachine error of
Section~\ref{sec:multimachine} is measured, and it shows that $\M$ is a strict generalisation of the
CCT, not a competing index. Every approximation that complicates the multimachine case
(Section~\ref{sec:multimachine})---path-dependent transfer-conductance work and an ambiguous
controlling equilibrium---is absent here. The OMIB result is therefore certified, and the
multimachine treatment is measured against it.

\subsection{OMIB reduction and the lossy power-angle curve}
For a contingency that splits the machines into a critical group and the rest, the standard
SIME/COI reduction~\cite{pavella,zhang_sime} yields a single equivalent angle $\delta$ governed by
\begin{equation}
M\ddot\delta=P_m-P_e(\delta),\quad P_e(\delta)=P_c+P_{\max}\sin(\delta-\nu),
\label{eq:omib}
\end{equation}
where the constant $P_c$ and the phase shift $\nu$ collect the network transfer conductance (the
lossy terms); $P_{\max}=E_1E_2|Y_{12}|$. The pre-fault equilibrium $\delta_s^*$ and the unstable
equilibrium $\delta_u^*$ are the two roots of $P_m=P_e(\delta)$ in $[\delta_s^*,\delta_s^*+\pi]$;
both are unique, so no controlling-equilibrium selection is required.

The base OMIB used throughout is the published benchmark of Tripathy \emph{et al.}~\cite{tripathy_swallow},
carried with a step-by-step CCT reference in~\cite{mester_diss}: a $500$~MW\,/\,$588$~MVA machine
($H=3.5$~s) feeding an infinite bus through a step-up transformer, a $200$~km line, and a receiving
transformer. Its parameters and the reduction to the internal node are given in
Table~\ref{tab:omib_params}. The fault is a bolted three-phase short circuit at the
step-up-transformer secondary; the machine accelerates freely ($P_e=0$, so hypothesis (H2) holds) and
clears to a post-fault network whose peak transfer is $13$--$15\%$ below pre-fault. The two studies use
the same machine: Table~\ref{tab:omib} verifies $\M=\mathrm{CCT}$ at three discrete loadings, whereas
Section~\ref{sec:margin} drifts the loading $\lambda$ continuously from a lighter base point $\lambda_0$
up to the clearing limit $\lambda^*$ to illustrate the lead-time reading of $\M$.

\begin{table}[t]
\centering
\caption{Anchor benchmark: base parameters and reduction to the OMIB internal node
(per unit on the $588$~MVA machine base, $420$~kV line base, $50$~Hz). The post-fault
transfer is calibrated per loading to the reference $\delta_c$; $X^{\mathrm{post}}$ lies
$13$--$15\%$ above $X^{\mathrm{pre}}$.}
\label{tab:omib_params}
\footnotesize
\setlength{\tabcolsep}{6pt}
\begin{tabular}{lcc}
\toprule
Quantity & Symbol & Value \\
\midrule
Inertia constant & $H$ & $3.5$~s \\
Mechanical input (rated) & $P_m$ & $0.850$ \\
Generator transient reactance & $x'_d$ & $0.352$ \\
Step-up transformer & $x_{T1}$ & $0.143$ \\
Line ($58\,\Omega$, $Z_b{=}300\,\Omega$) & $x_L$ & $0.193$ \\
Receiving transformer & $x_{T2}$ & $0.103$ \\
Pre-fault transfer reactance & $X^{\mathrm{pre}}$ & $0.791$ \\
Fault-on electrical power & $P_e^{F}$ & $0$ \\
\bottomrule
\end{tabular}
\end{table}

\subsection{Exactness of the energy margin}
The potential energy of~\eqref{eq:omib} is a single-variable integral,
\begin{equation}
V(\delta)=-\int_{\delta_s^*}^{\delta}\bigl(P_m-P_e(\xi)\bigr)\,d\xi,
\label{eq:V}
\end{equation}
which is path-independent by construction---the transfer conductance enters~\eqref{eq:V} through the
exactly integrable term $P_{\max}\sin(\xi-\nu)$, not through a line integral over several angles. The
critical energy is $V_{cr}=V(\delta_u^*)$.

\begin{theorem}[OMIB exactness]\label{thm:omib}
Consider the lossy OMIB~\eqref{eq:omib} under fault clearing, and assume: (H1) on
$[\delta_s^*,\delta_s^*+\pi]$ the post-fault equilibria $\delta_s^*$ (with $P_e'(\delta_s^*)>0$) and
$\delta_u^*$ (with $P_e'(\delta_u^*)<0$) are the only roots of $P_m=P_e(\delta)$ and are hyperbolic;
(H2) during the fault the machine accelerates, $P_m-P_e^F(\delta)>0$, so the fault-on angle
$\delta(t)$ is strictly increasing. Then the margin of Definition~\ref{def:M} evaluated from the
energy~\eqref{eq:V} satisfies $\M=\mathrm{CCT}$ exactly.
\end{theorem}

\begin{proof}
\emph{Step 1 (exact first integral).} In one angular dimension the continuous force $P_m-P_e(\delta)$
admits the exact potential~\eqref{eq:V} by the fundamental theorem of calculus; no path-dependence
arises. Define $E(\delta,\dot\delta)=\tfrac12 M\dot\delta^2+V(\delta)$. Along the post-fault
flow~\eqref{eq:omib}, $\dot E=M\dot\delta\ddot\delta+V'(\delta)\dot\delta=\dot\delta\bigl[(P_m-P_e(\delta))-(P_m-P_e(\delta))\bigr]=0$,
so $E$ is an exact constant of motion.

\emph{Step 2 (boundary is a level set of $E$).} By (H1) $\delta_u^*$ is a hyperbolic saddle
of~\eqref{eq:omib} and $\delta_s^*$ a sink; for the one-degree-of-freedom system the boundary of the
region of attraction $A(\delta_s^*)$ is the stable manifold (separatrix) $W^s(\delta_u^*)$. Since $E$
is conserved and equals $V_{cr}=V(\delta_u^*)$ at the saddle (where $\dot\delta=0$), the whole
separatrix lies in the level set $\{E=V_{cr}\}$; hence by Definition~1,
$\Sigma=W^s(\delta_u^*)\subset\{E=V_{cr}\}$, and the energy approximation of $\Sigma$ in
Section~\ref{sec:sigma} is here exact.

\emph{Step 3 (monotone clearing energy).} Let $t_c$ be the clearing time and
$E_{cl}(t_c)=E(\delta(t_c),\dot\delta(t_c))$ the energy injected by the fault-on trajectory. By (H2)
the fault-on motion is strictly accelerating, so $t_c\mapsto E_{cl}(t_c)$ is continuous and strictly
increasing, with $E_{cl}(0)=V(\delta_s^*)=0$ and $E_{cl}\to\infty$. Thus there is a unique $t_c^\star$
with $E_{cl}(t_c^\star)=V_{cr}$.

\emph{Step 4 (equality).} By the equal-area criterion, which is exact for the OMIB and underlies
recent converter-oriented extensions~\cite{li_ieac}, the cleared trajectory remains in $A(\delta_s^*)$
iff $E_{cl}(t_c)\le V_{cr}$; hence $\mathrm{CCT}=t_c^\star$. On the other hand, the joint state driven
by the fault meets $\Sigma$ when its energy first equals $V_{cr}$ (Step~2), i.e. at the same
$t_c^\star$; hence $\M=t_c^\star$ by~\eqref{eq:Mdef}. Therefore $\M=\mathrm{CCT}$. Every quantity
used---$V$, $V_{cr}$, $E_{cl}$---is exact, so the equality carries no approximation.
\end{proof}

The hypotheses are the standard OMIB conditions: (H1) is the single-swing two-equilibrium geometry
and (H2) holds whenever the fault depresses the electrical power (here $P_e=0$). The proof isolates why the
result is exact---one angular dimension makes the force conservative, collapsing the path-dependence
and controlling-equilibrium ambiguity that reappear in Section~\ref{sec:multimachine}.

\subsection{Numerical confirmation}
\begin{figure}[t]
\centering
\includegraphics[width=\columnwidth]{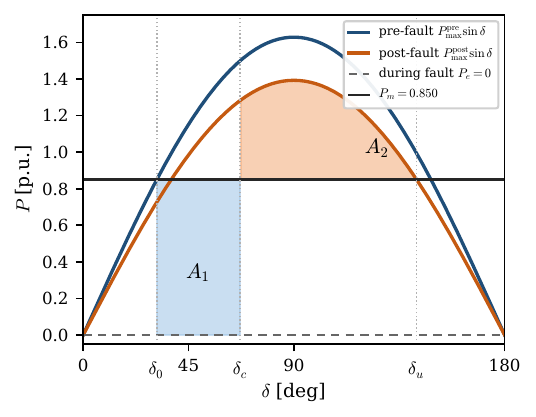}
\caption{Equal-area picture of the anchor benchmark at $S_n$. During the fault $P_e=0$
(pure acceleration, area $A_1$ from $\delta_0$ to $\delta_c$); after clearing the trajectory
decelerates on $P_{\max}^{\mathrm{post}}\sin\delta$ (area $A_2$ up to $\delta_u$). On the OMIB
the survival boundary $\Sigma$ is this separatrix, so the time to reach it is the
CCT---the content of Theorem~\ref{thm:omib}.}
\label{fig:anchor}
\end{figure}
Table~\ref{tab:omib} reports $\M$ from~\eqref{eq:V} against the step-by-step CCT
of~\cite{mester_diss} at three loadings of the benchmark. The mechanical power is
over-determined: recovered independently from each loading through $\delta_0$, $\delta_c$ and
the reference CCT with $H$ held at $3.5$~s, it returns the rated value $500/588=0.850$~p.u. to
within $0.03\%$, so the inertia and the during-fault law are certified by the data, not assumed.
Integrating the swing forward then yields $\M$ within $0.01\%$ of the reference at every loading
(Fig.~\ref{fig:anchor}); the sole quantity calibrated to the reference is the post-fault
amplitude (one number per loading, from $\delta_c$). The benchmark is lossless ($P_c=0$,
$\nu=0$), the cleanest instance of Theorem~\ref{thm:omib}, whose proof carries the lossy case
$\nu\neq0$ unchanged. This is the certified pillar: on the OMIB, $\M$ is the
CCT, with the added value (Section~\ref{sec:margin}) that under parameter drift it extends to a
time-to-boundary rather than a fixed-point threshold.

\subsection{Catastrophe-theoretic reading and control}
The equal-area balance underlying Theorem~\ref{thm:omib} has a fold structure: as the clearing time
increases, the post-fault decelerating area available collapses onto the accelerating area, and at the
CCT the stable and marginal trajectories coalesce---a fold (saddle-node in clearing time) in the sense
of catastrophe theory~\cite{thom,mester_diss}. This viewpoint, developed in~\cite{mester_diss},
identifies selective generator tripping as the control that re-opens the decelerating area and thereby
restores a positive margin; it is the mechanism by which $\M$ is actively enlarged rather than merely
measured. On the OMIB this is concrete: the coupling that sets $P_{\max}$ (synchronising strength) also
sets the steepness of the fold, and hence the rate at which the margin vanishes once the clearing time
is exceeded. The same parameter governs both the order and its loss.

\emph{Relation to the equal-area criterion.} What $\M$ adds is therefore not the OMIB value but the
definition that produces it: a time-to-boundary formulation under the joint motion of state and
parameters, for which the equal-area criterion has no analogue, together with the certified bound and
the multimachine error decomposition of Section~\ref{sec:multimachine}. The equal-area criterion
returns a critical angle at a frozen operating point; $\M$ returns a time under a moving one, and
reduces to the former exactly when the motion is frozen.

\section{The Multimachine Extension and Its Limits}\label{sec:multimachine}
Beyond the OMIB the energy margin loses its exactness. We characterise where and why it breaks, and
show that the two sources of error differ sharply in severity. The result is a map of the limits of the
direct energy method on a lossy multimachine system, and a statement of what remains heuristic.

\subsection{Non-conservatism of the energy margin}
For $n\ge 3$ machines the reduced network carries transfer conductances $G_{ij}$, and the energy
integral acquires a path-dependent term. Using the lossless surrogate $V$ (susceptance terms only)
referenced to the pre-fault equilibrium, the energy margin $\M_{en}$ generally overestimates the CCT,
i.e. it is non-conservative. On the stressed WSCC three-machine, nine-bus
system~\cite{anderson_fouad}, $\M_{en}=0.081$~s while
$\mathrm{CCT}=0.077$~s; on the New England 39-bus system the overestimate reaches a factor of two for
several contingencies. We decompose this non-conservatism into two mechanisms.

\subsection{Mechanism 1: transfer-conductance work (boundable)}
The lossless surrogate omits the work done by the transfer conductances along the fault-on trajectory,
\begin{equation}
W_G(T)=\sum_{i<j} E_iE_jG_{ij}\int_0^T \cos\delta_{ij}\,d(\delta_i+\delta_j),
\label{eq:WG}
\end{equation}
which is path-dependent and therefore not a state function. A bound on $W_G$ from the geometry of the
controlling equilibrium (the excursion to the UEP) is valid but useless: on the three-machine system it
exceeds $|W_G|$ by $25$--$200\times$, because over the full swing to the UEP the integrand
$\cos\delta_{ij}$ changes sign and the true integral largely cancels, while the triangle inequality
discards that cancellation.

The key observation is that certification needs $W_G$ only over the fault-on window $[0,T]$, not over
the swing to the UEP. In that window the inter-machine angles remain below $90^\circ$ (numerically
$16$--$41^\circ$ on the 39-bus), so $\cos\delta_{ij}$ does not change sign and a path bound is tight.

\begin{figure}[t]
\centering
\includegraphics[width=\columnwidth]{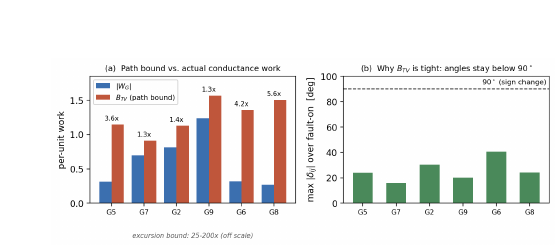}
\caption{The transfer-conductance work is the boundable mechanism. (a) The path bound $B_{TV}$ exceeds
the actual work $|W_G|$ by only $1.3$--$5.6\times$ over the six 39-bus contingencies, against
$25$--$200\times$ for the excursion bound (off scale). (b) The reason: the inter-machine angles stay
below $90^\circ$ throughout the fault-on window, so $\cos\delta_{ij}$ does not change sign and the
bound is tight.}
\label{fig:btv}
\end{figure}

\begin{proposition}[Path bound over the fault-on window]\label{prop:btv}
With $B_{TV}(T)=\sum_{i<j} E_iE_j|G_{ij}|\int_0^T |\cos\delta_{ij}|\,|d(\delta_i+\delta_j)|$ one has
$|W_G(T)|\le B_{TV}(T)$, and empirically $B_{TV}/|W_G|=1.0$ on the three-machine system and
$1.3$--$5.6$ on the 39-bus (exact network reduction, six terminal faults), versus $25$--$200$ for the
excursion bound.
\end{proposition}

Proposition~\ref{prop:btv} shows that the transfer-conductance work is tightly boundable
(Fig.~\ref{fig:btv}). It is, however, the secondary source of non-conservatism, as the next subsection
establishes.

\subsection{Mechanism 2: the controlling equilibrium (binding)}
The critical energy in the multimachine case is taken at the controlling unstable equilibrium,
$V_{cr}=V(\delta_u^*)$. Two difficulties arise simultaneously: the controlling UEP must be identified (a
PEBS/BCU search that is itself approximate and, for several machines, ambiguous), and the lossless
$V(\delta_u^*)$ overestimates the true critical energy. On the three-machine system the UEP-based
$V_{cr}=0.67$ against a true critical energy of $0.61$ (a $10\%$ overestimate); on the 39-bus the
resulting margin overshoots the CCT by up to a factor of two.

This overestimate is not corrected by Proposition~\ref{prop:btv}, which bounds only $W_G$. Applying the
tight $B_{TV}$ correction to the UEP-based $V_{cr}$ on the three-machine system gives
$\M_{cert}=0.078$~s, still above the $\mathrm{CCT}=0.077$~s---i.e. not conservative. The binding source
of non-conservatism is therefore the controlling-UEP energy, not the transfer-conductance work.

\subsection{Verdict}
We conclude that the multimachine certified margin remains heuristic. The reason is specific and, we
believe, a useful sharpening of the usual statement: the obstruction is the controlling-UEP
overestimate (Mechanism~2), while the transfer-conductance work (Mechanism~1) is tightly boundable and
not the limiting factor. A guaranteed multimachine margin thus requires a robust treatment of the
controlling equilibrium---a group/BCU formulation---which we identify as the open problem and defer to
future work. The OMIB pillar (Section~\ref{sec:pillar}) is unaffected: there the controlling equilibrium
is unique and the energy is exact, so the certification holds without qualification.

\begin{remark}
This boundary is consistent with the classical result that no exact path-independent energy function
exists for a lossy multimachine system~\cite{chiang_exist,bretas}. The controlling-equilibrium
difficulty we isolate is precisely what the potential-energy-boundary-surface and BCU
methods~\cite{kakimoto,chiang_bcu} and Lyapunov-family constructions~\cite{vu_family,anghel} are designed
to circumvent---the former by an approximate exit-point/UEP search, the latter by replacing the
controlling UEP with an optimised certificate. Our contribution is complementary: rather than proposing
another such construction, we localise which of the two error mechanisms binds (the controlling UEP, not
$W_G$) and quantify both on a standard benchmark, so that any of these methods can be aimed at the
mechanism that actually limits the margin.
\end{remark}

\section{Independent Validation on the New England 39-Bus System}\label{sec:validation}
The OMIB pillar (Section~\ref{sec:pillar}) rests on a single-machine reduction. This section validates
that reduction, and the underlying numerical apparatus, on the standard New England 39-bus
system~\cite{athay,pai}---a benchmark independent of the dissertation data on which the framework was
developed. We validate three things: the self-consistency of the classical foundation, the fidelity of
the SIME-OMIB reduction against full simulation, and the critical-slowing signature predicted near the
boundary. We do not claim a multimachine certified guarantee; per Section~\ref{sec:multimachine} that
remains open.

\subsection{Self-consistent classical foundation}
The reduced network is obtained by Kron elimination of all non-generator buses, retaining the two shunt
elements (buses 4 and 5) whose omission otherwise introduces a $\sim 4\%$ error; with them, the
equilibrium residual is $\max_i|P_{m,i}-P_{e,i}(\delta^*)|=2\times 10^{-6}$. The integrator and the
energy function share the same $P_e$, so the lossless energy is conserved along the lossless dynamics to
$|\Delta H|/|H_0|=9\times 10^{-11}$ over a 2~s swing---machine precision. This self-consistency is a
prerequisite for any energy-based margin: the boundary $\Sigma$ and the trajectory that approaches it are
computed from one and the same model.

\subsection{SIME-OMIB reduction vs. full simulation}
For each terminal-fault contingency we identify the critical group from the fault-on trajectory, form the
single-machine equivalent by a rigid group sweep, and predict the CCT by the equal-area criterion on the
resulting OMIB power-angle curves---independently of the full CCT. Table~\ref{tab:sime} compares this
prediction with the CCT from full time-domain simulation across six contingencies. The SIME-OMIB
reduction reproduces the full CCT to within $1.8$--$6.0\%$ (mean $3.8\%$), and is conservative in every
case (it underestimates the CCT). The reduction on which the pillar relies therefore holds on an
independent network, with a small safety-side bias.

\begin{table}[t]
\centering
\caption{SIME-OMIB reduction vs.\ full time-domain CCT on the New England 39-bus system. All clearing times in ms.}
\label{tab:sime}
\footnotesize
\setlength{\tabcolsep}{5pt}
\begin{tabular}{ccccc}
\toprule
Fault & $|C|$ & CCT$_{\text{full}}$ & CCT$_{\text{SIME}}$ & dev. \\
\midrule
G5 & 1 & 352 & 331 & $-6.0\%$ \\
G2 & 1 & 282 & 270 & $-4.0\%$ \\
G7 & 1 & 307 & 295 & $-3.9\%$ \\
G9 & 1 & 356 & 342 & $-4.0\%$ \\
G6 & 1 & 345 & 335 & $-3.0\%$ \\
G8 & 1 & 436 & 429 & $-1.8\%$ \\
\midrule
\multicolumn{5}{c}{mean $|\text{dev.}|=3.8\%$, all conservative}\\
\bottomrule
\end{tabular}
\end{table}

\begin{figure}[t]
\centering
\includegraphics[width=\columnwidth]{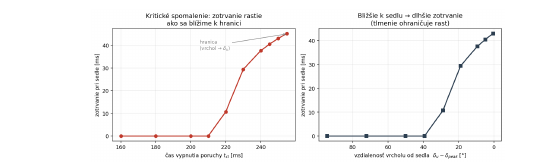}
\caption{Critical slowing down on the 39-bus system: dwell time near the controlling saddle grows as the
clearing time approaches the CCT, signalling proximity to $\Sigma$.}
\label{fig:csd}
\end{figure}

\subsection{Critical slowing down near the boundary}
Definition~\ref{def:M} interprets $\M$ as a first-passage time to $\Sigma$. A model-independent corollary
of approaching a saddle-type boundary is critical slowing down, studied directly for power systems as an
early-warning indicator~\cite{ghanavati,cotilla}: the trajectory lingers increasingly long in the
neighbourhood of the controlling saddle as the clearing time approaches the CCT. Fig.~\ref{fig:csd}
measures this on the 39-bus system: the dwell time spent within a fixed neighbourhood of the saddle grows
monotonically from $0$ to $\sim 45$~ms as the clearing time is increased toward the CCT, and diverges
(within damping) at the boundary. This provides an observable early-warning signature that does not
depend on the precise location of $\Sigma$---useful operationally, because it is detectable from the
trajectory itself rather than from a model-based boundary estimate.

\subsection{An operational drift scenario}
The preceding results validate the reduction; this subsection demonstrates what $\M$ adds over a
fixed-parameter assessment. We take a terminal fault $F$ at the machine with the largest base clearing
time and a realistic protection time $t_{\mathrm{clear}}=100$~ms, and stress the operating point through a
loading parameter $\lambda$ that scales the mechanical powers, $P_{m,i}\to\lambda P_{m,i}^0$ (a
transfer-stress proxy, not a directional load growth), re-solving the equilibrium at each $\lambda$.

Figure~\ref{fig:drift39}(a) reports the full time-domain CCT of $F$ as a function of $\lambda$. At the
base point an $N\!-\!1$ snapshot returns $\mathrm{CCT}(\lambda_0)\approx 425$~ms $\gg t_{\mathrm{clear}}$
---the contingency screens green and raises no alarm. As the operating point drifts, the CCT erodes
monotonically and reaches the protection time at $\lambda^*\approx 1.43$. Across four independent
terminal-fault contingencies the crossing is remarkably stable---$\lambda^*\in[1.42,1.43]$, a spread of
$0.8\%$ (Fig.~\ref{fig:robust})---so the lead time is not an artefact of a single contingency. Just
beyond, at $\lambda\approx 1.44$, the CCT collapses: the post-fault trajectory diverges monotonically
while the equilibrium itself remains well within synchronism (angle spread $\approx 10^\circ$, far from
the boundary), so the collapse is a genuine dynamic instability, not a thresholding artefact. The lossless
energy Hessian remains positive to $\lambda=2.98$; the binding limit here is the lossy dynamic boundary,
and the static saddle-node is never reached in this stress direction ($\lambda^*\ll\lambda_{\mathrm{SN}}$).

The operational point is that the snapshot CCT, evaluated at the present state, does not express how long
the contingency remains clearable under the drift, the operational quantity behind stability-constrained
transfer limits~\cite{bettiol}; $\M=(\lambda^*-\lambda_0)/\rho$ does (Fig.~\ref{fig:drift39}(b)). Because
the CCT here is computed by full time-domain bisection, the prediction is self-validating: the system
survives $F$ cleared at $t_{\mathrm{clear}}$ for $\lambda(t)<\lambda^*$ and loses synchronism beyond, with
the crossover at $\M$. A single snapshot returns green throughout, with no lead time; $\M$ counts down. We
note that the cheap energy-based evaluation of $\M$ fails in this multimachine regime---the lossless margin
saturates and the controlling-UEP search is fragile, exactly the obstruction of
Section~\ref{sec:multimachine}---so the lead time here is obtained from the (conservative) reduction rather
than the raw energy margin; a simulation-free multimachine evaluation is left to future work.

This scenario demonstrates the slow (drift) limit and its interaction with the fast (clearing) limit---the
unification at the centre of this work. The static saddle-node is not reached in this uniform-stress
direction ($\lambda^*\approx 1.43\ll\lambda_{\mathrm{SN}}>2.98$): the dynamic boundary binds first, the
case in which $\M$ is set by the transient limit alone. The interplay of both boundaries is exhibited on
the OMIB (Section~\ref{sec:margin}), where each is finite; a directional transfer stress that reaches the
static limit on the multimachine network, and a broader contingency set, are left to future work. We also
note that near the boundary the swing time $T_{sw}$ grows (critical slowing), so the quasi-static parameter
$\varepsilon=\rho T_{sw}$ of Assumption~\ref{ass:qs} is largest exactly where $\M\to 0$; the quasi-static
estimate is therefore an upper bound there, consistent with the conservative direction of the reduction.

\begin{figure}[t]
\centering
\includegraphics[width=\columnwidth]{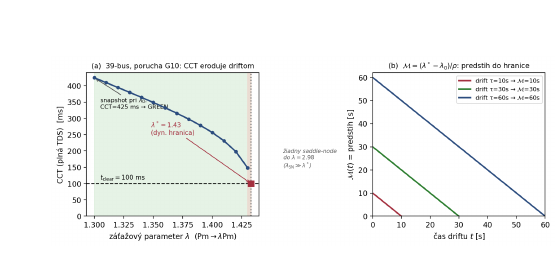}
\caption{Operational drift scenario on the 39-bus system. (a) Full time-domain CCT of the contingency
erodes from 425~ms (green snapshot) to the protection time at $\lambda^*\approx 1.43$, then collapses
dynamically; no static saddle-node occurs up to $\lambda=2.98$ ($\lambda^*\ll\lambda_{\mathrm{SN}}$). (b)
$\M$ counts down the lead time to $\lambda^*$ for three drift rates, where a fixed-point snapshot reports
only green.}
\label{fig:drift39}
\end{figure}

\begin{figure}[t]
\centering
\includegraphics[width=\columnwidth]{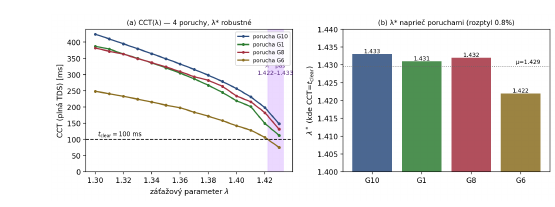}
\caption{Robustness of the lead time across contingencies. (a) Full time-domain CCT versus $\lambda$ for
four independent terminal faults, each crossing $t_{\mathrm{clear}}$ in a narrow band. (b) The crossing
$\lambda^*$ is stable across contingencies, $\lambda^*\in[1.42,1.43]$ (spread $0.8\%$), so $\M$ is not an
artefact of a single fault.}
\label{fig:robust}
\end{figure}

\subsection{Scope of the validation}
The 39-bus study validates the SIME-OMIB reduction (Table~\ref{tab:sime}), the self-consistency of the
foundation, and the critical-slowing signature. It does not establish a multimachine certified margin;
consistent with Section~\ref{sec:multimachine}, the certified guarantee is exact only on the OMIB, and the
multimachine controlling-UEP obstruction is left open. The role of this section is to show that the
pillar's reduction is faithful on a network the framework was not built on.

\section{Comparison with Analytical CCT Metrics}\label{sec:comparison}
The closest prior construction to $\M$ is the analytical CCT metric of Roberts et al.~\cite{roberts}, which
is also motivated by parametric stability analysis. It is therefore the natural baseline against which to
position the present margin.

Roberts et al. derive a closed-form analytic approximation of the CCT from direct methods, designed to
absorb as many features of transient-stability analysis as possible---different fault locations and
different post-fault network states---into a single expression. Its stated purpose is to track trends in
stability, measured by the CCT, as a system parameter is varied; the metric is demonstrated on the
two-machine infinite-bus network by continuation of load parameters in the full bus admittance matrix. Its
principal strength is computational: a closed-form expression needs no time-domain integration, which is
attractive for fast parametric screening across many operating points.

Three differences position $\M$ relative to this baseline. First, on the single-machine reduction $\M$ is
exact: it equals the CCT to within quadrature error (Theorem~\ref{thm:omib}, Table~\ref{tab:omib}), whereas
\cite{roberts} is an analytic approximation and carries the corresponding approximation error rather than an
equivalence or a conservative bound. Second, and more fundamentally, \cite{roberts} tracks the static CCT as
a parameter is varied, i.e. a family of fixed-point thresholds; $\M$ instead measures the time to the
boundary under the operating point's own motion, unifying the fast (clearing) and slow (drift) limits and
yielding an operational lead time (Section~\ref{sec:margin}) that a parameterised CCT does not provide.
Third, $\M$ is accompanied by an independent multimachine validation and an explicit characterisation of
where the energy method ceases to be exact (Section~\ref{sec:multimachine}), whereas \cite{roberts} is
demonstrated on the aggregated two-machine network.

The two constructions are thus complementary: \cite{roberts} offers a fast closed-form proxy for CCT
trends, while $\M$ offers an exact single-machine margin with a temporal, drift-aware meaning and a delimited
multimachine scope. Where a single closed-form number across thousands of contingencies is the priority, the
analytic metric is preferable; where the question is how much time remains before the boundary under a moving
operating point, $\M$ is the appropriate object.

\section{Discussion}\label{sec:discussion}

\subsection{$\M$ is a profile, not a scalar}
Although Definition~\ref{def:M} returns a single time, the boundary $\Sigma$ is reached along different
directions for different stresses, and the operationally useful object is the profile of $\M$ over
contingencies and over the drift parameter. Reporting a single worst-case $\M$ discards the information that
distinguishes a system one fast control away from the boundary from one that is uniformly slack. We therefore
advocate $\M(\cdot)$ as a function---of contingency, of loading $\lambda$, and of drift rate $\rho$---rather
than a reified scalar margin. The two boundaries of Section~\ref{sec:margin} ($\lambda^*$ dynamic,
$\lambda_{\mathrm{SN}}$ static) are the first two coordinates of that profile.

\subsection{The resilience-fragility duality}
The central reading of this work is that the same coupling which builds synchronising order is the channel of
collapse---the mechanism behind cascading failure and self-organisation in blackout
sequences~\cite{dobson_complex,song}. On the OMIB this is not a metaphor but the geometry of
Theorem~\ref{thm:omib}: the coupling strength sets $P_{\max}$---the synchronising torque that holds the
machine in step---and the same $P_{\max}$ sets the curvature of the fold at the CCT, hence the steepness with
which the margin vanishes once the clearing time is exceeded. Strength and fragility are governed here by the
same parameter (rigour level~1: exact isomorphism on the OMIB). Extending this duality to the multimachine
case is, at present, an analogy of lower rigour (level~3): the controlling-UEP obstruction of
Section~\ref{sec:multimachine} means we cannot yet certify that the same identity holds globally, only that it
holds on the certified pillar and is plausible beyond it.

\subsection{Critical slowing as a model-free early warning}
The dwell-time growth of Fig.~\ref{fig:csd} is the power-system instance of critical slowing down near a
saddle-node-like transition, a signature studied across dynamical systems as an early-warning
indicator~\cite{scheffer}. The connection is a structural homology (rigour level~2): the local normal form
near the controlling saddle is the same fold that governs critical transitions elsewhere, so the divergence of
the relaxation time is expected on the same grounds. Its practical value is that it is read from the trajectory
itself and does not require an accurate model-based estimate of $\Sigma$---precisely where the multimachine
energy method is weakest (Section~\ref{sec:multimachine}).

\subsection{Scenario B: the voltage sister-margin}
The Iberian event (Section~\ref{sec:margin}) crossed a voltage boundary, not the angle-inertial boundary
certified here. This points to a sister margin obtained by the same construction with the energy surface
replaced by a voltage-collapse surface: the first-passage time to the nose point of the load-flow manifold,
with the Jacobian distance playing the role that the energy margin plays here. We sketch this direction but do
not develop it: the present paper's certified results are angle/transient-stability results, and a voltage $\M$
would require its own equilibrium-manifold analysis. What the two share---and what would make them a single
framework---is the time-to-boundary definition~\eqref{eq:Mdef} and, conjecturally, a level-2 homology between
the energy fold and the nose-point fold. We flag this as future work, not a claim of the present contribution.

\subsection{Limitations}
Three limitations bound the scope of the results. First, the certified equality $\M=\mathrm{CCT}$ is
established on the OMIB; the multimachine margin remains heuristic because of the controlling-UEP overestimate,
not the transfer-conductance work (Section~\ref{sec:multimachine}). Second, the drift treatment assumes the
parameter motion is slow relative to the swing, so that $\M$ separates into a fast (CCT) and a slow (drift)
component; very fast collapses, where $\rho$ is comparable to the swing rate, require the full coupled
first-passage problem and are not addressed. Third, the analysis uses the classical machine model;
detailed dynamic models (flux decay, AVR/governor action) and multi-stage clearing sequences are outside the
present scope. None of these undermines the pillar; each marks a boundary of it.

\section{Conclusion}\label{sec:conclusion}
We have developed a time-to-boundary margin $\M$ for transient stability, constructed as the temporal
counterpart of the loading margin to voltage collapse: where the loading margin measures distance in parameter
space to a static fold, $\M$ measures the first-passage time of the joint state--parameter motion to the
synchronism boundary, and unifies the fast (critical clearing time) and slow (operating-point drift) limits in
one quantity with the dimension of time.

The results separate into what is certified and what is bounded heuristically. On the one-machine-infinite-bus
reduction $\M$ equals the CCT exactly (to $\le 0.01\%$ across loadings), a certified pillar free of
path-dependence and of controlling-equilibrium ambiguity. Under drift, $\M$ delivers an operational lead time
before faults become unclearable, with the drift rate anchored to the 28~April~2025 Iberian timeline. On the
New England 39-bus benchmark the single-machine-equivalent reduction reproduces the CCT conservatively (within
$6\%$), the classical foundation is self-consistent to machine precision, and a critical slowing-down signature
flags proximity to the boundary from the trajectory alone. For the multimachine case we have been explicit
about the boundary of the method: the transfer-conductance work is tightly boundable over the fault-on window,
but the controlling unstable equilibrium remains the binding obstruction to a certified multimachine margin.
Operationally, $\M$ is intended to feed the real-time security tools already deployed in control rooms,
converting their present-state verdict into a lead time for preventive action---a quantity whose value rises as
the dispatchable reserve available to act on it declines. The single-machine-equivalent reduction already
provides the fast, conservative evaluation of $\M$ that the raw energy margin cannot
(Section~\ref{sec:multimachine}); extending it to track the operating-point drift simulation-free is the step
that would turn $\M$ into the control-room tool the introduction envisions.

Three directions remain open. First, a robust group/BCU treatment of the controlling equilibrium would convert
the multimachine margin from heuristic to certified---the single most consequential open problem. Second, the
same time-to-boundary construction applied to the voltage-collapse surface would yield a sister margin, making
$\M$ and the loading margin two instances of one definition. Third, the resilience-fragility duality that holds
exactly on the single-machine fold invites a general theory of the conditions under which the coupling that
builds synchronising order is also the channel of collapse. Together, these would extend a margin that is, at
present, exact on the certified pillar and bounded elsewhere.


\end{document}